\documentclass[envcountsame]{llncs}

\usepackage{amsmath,amssymb,amsfonts}
\usepackage[dvipdfmx]{graphicx}
\usepackage{listings}
\usepackage{xcolor}
\usepackage{url}
\usepackage[caption=false]{subfig}
\usepackage{multirow}
\usepackage{multicol}
\usepackage[linesnumbered,ruled,vlined]{algorithm2e}
\usepackage{mathtools}
\usepackage{wrapfig}
\usepackage{booktabs}
\usepackage{here}
\usepackage{bussproofs}

\usepackage{article}

\begin{document}

\mainmatter

\title{A Hypergraph-based Formalization of Hierarchical Reactive Modules   and a Compositional Verification Method}

\author{
    Daisuke Ishii
}
\institute{
    Japan Advanced Institute of Science and Technology, Ishikawa, Japan \\
    \email{dsksh@jaist.ac.jp} 
}

\maketitle

\begin{abstract}
  The compositional approach is important for reasoning about large and complex systems.
  In this work, we address synchronous systems with hierarchical structures, which are often used to model cyber-physical systems.
  We revisit the theory of reactive modules and reformulate it based on hypergraphs to clarify the parallel composition and the hierarchical description of modules.
  Then, we propose an automatic verification method for hierarchical systems.
  Given a system description annotated with assume-guarantee contracts,
  the proposed method divides the system into modules and verifies them separately to show that the top-level system satisfies its contract.
  Our method allows an input to be a circular system in which submodules mutually depend on each other.
  Experimental result shows our method can be effectively implemented using an SMT-based model checker.
  %
\end{abstract}

\section{Introduction}

\emph{Synchronous reactive systems} are a basic model used in the design of \emph{cyber-physical systems (CPSs)}~\cite{alurSynchronous2015},
which are typically described as a feedback loop model with plant and digital controller modules.
Although it is important for industrial CPS products to formally verify the safety of their models,
the effort has not been sufficient. The scale and complexity of the models are hampered by the expertise and computational complexity required by formal method tools.

A divide-and-conquer approach could be the cure for scalability.
The theory of reactive modules~\cite{alurReactive1999,alurSynchronous2015} provides a foundation of the compositional reasoning~\cite{giannakopoulouCompositional2018}.
It formalizes a system as a set of \emph{modules} (or agents or components) that behave synchronously on a sequence of rounds and enables to verify the \emph{implementation} (or {refinement}) relation between composite modules.
In the verification, the \emph{assume-guarantee rule}~\cite{alurReactive1999} is crucial to reason about circular systems in which submodules depend on each other. 
Although it has been studied for decades, the theory is still underutilized in practice due to 
its discrepancy from the actual CPS descriptions and the lack of automated verification methods.

In this paper, we consider the verification of synchronous system models that is composed of reactive modules $M_1$,\ldots, $M_n$.
Assuming that each module $M_j$ is given a contract $(M_{a.j},M_{g.j})$ and satisfies it (we denote the fact by $M_j \PC M_{a.j} \Impl M_{g.j}$), we verify that the top-level system $M_1 \PC \cdots \PC M_n$ satisfies its contract
$(M_{a}, M_{g})$.
As is the case in \cite{alurMOCHA1998,bostromContractbased2016,dragomirRefinement2020}, verification based on the assume-guarantee contracts requires an interactive proof process~.
On the other hand, an automatic method~\cite{championCoCoSpec2016,championKIND2016} is proposed that abstracts submodules using their contracts and efficiently performs model checking.
However, because of the abstraction, the method has the disadvantage of finding spurious counterexamples, especially when dealing with circular systems.

The objective of this research is to bridge the gap between the theory of reactive modules and the hierarchical design of practical systems, and to propose an automated compositional verification method based on the theory.
Our contributions are summarized as follows:
\begin{enumerate}
  \item We formalize the hierarchical structure commonly used in practical modeling languages, e.g. Lustre and Simulink, as a composition of reactive modules.
  Our formulation features the use of hypergraphs to describe hierarchical structures.
  We extend the definition of modules for hierarchization and show how it differs from parallel composition.
  \item We propose a compositional verification method 
  that transforms a hierarchical module $M[M_1,\ldots,M_n]$ into a form of a composition $M_1 \PC\cdots\PC M_n \PC M^\dagger$ and then checks that it satisfies a given contract $(M_a,M_g)$.
  We show how to validate the implementation relation
  $M[M_1,\ldots,M_n] \PC M_{a} \Impl M_{g}$ automatically.
  The effectiveness of the method is confirmed by applying our implementation to several examples.
\end{enumerate}

\begin{figure}[t!]
  \centering
  \begin{minipage}[b]{.5\textwidth}
    \centering
    \includegraphics[width=\textwidth]{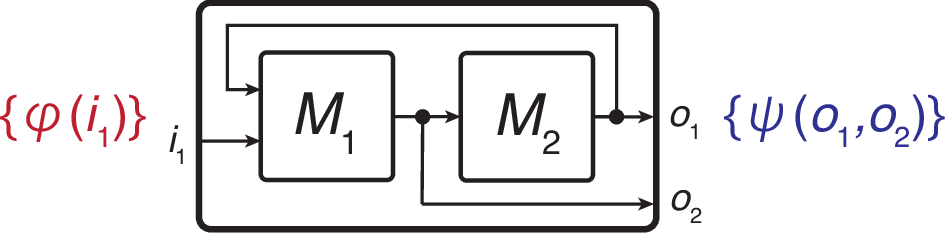} 
    \caption{Example system.}
    \label{f:intro}
  \end{minipage}
  \hspace{.5em}
  \begin{minipage}[b]{.45\textwidth}
    \centering
    \subfloat{\includegraphics[height=.04\textheight]{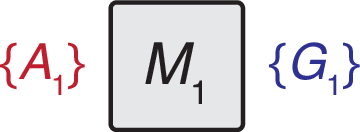}}
   \hspace{1em}
    \subfloat{\includegraphics[height=.04\textheight]{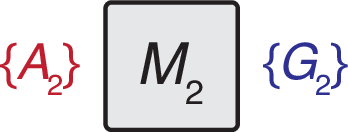}}

    \subfloat{\includegraphics[width=\textwidth]{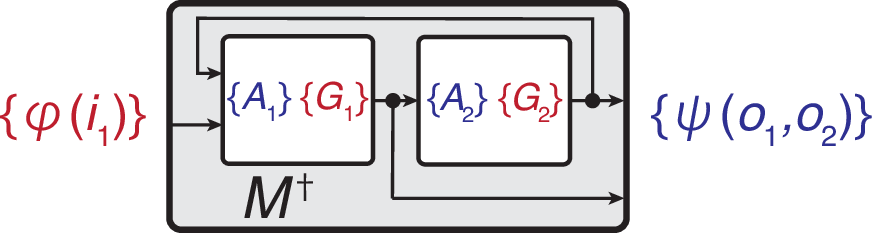}}
    \caption{Compositional verification.} 
    \label{f:intro:cv}
  \end{minipage}
\end{figure}

\noindent
\emph{Example}. 
Hierarchical modules can be illustrated as a flow diagram in Fig.~\ref{f:intro} in which rectangles represent modules e.g. $M_1$;
outer rectangle represents a hierarchical module $M[M_1,M_2]$.
Each module is equipped with input, output, and hidden state variables, and
is interpreted as a reaction relation between their values.
%
We consider a verification of that the module satisfies a contract $(\varphi(i_1),\psi(o_1,o_2))$, assuming the submodules are given sub-contracts $(A_1,G_1)$ and $(A_2,G_2)$.
A compositional verification can be done in three ways:
\begin{itemize}
  \item The method based on the reactive module theory regards $M[M_1,M_2]$ as a parallel composition $M_1 \PC M_2 \PC \cdots$ (if other modules are used, they should be composed together).
  Then, we deduce that the system satisfies the top-level contract by applying the inference rules to the assumptions, along with the composition structure.
  As the number of modules increases, the proof becomes more complex.
  \item The method based on abstraction~\cite{championKIND2016} regards each submodule $M_j$ as a reaction relation, e.g. $\Hist\,A_j \Rightarrow G_j$ described with a past fragment of LTL,
  and verifies $M[M_1,M_2]$ as a whole.
  Since there is a circular wiring in Fig.~\ref{f:intro}, this method may result in a spurious counterexample.
  \item The proposed method decomposes $M[M_1,M_2]$ as $M_1 \PC M_2 \PC M^\dagger$ (Fig.~\ref{f:intro:cv}) where the module $M^\dagger$ represents the top-level description content equipped with interface variables with $M_1$ and $M_2$.
  We formalize $M^\dagger$ using hypergraphs and propose how to properly give sub-contracts to it.
  Then, we show that if the verification for the submodules $M_1$, $M_2$ and $M^\dagger$ succeeds, then the verification goal for $M[M_1,M_2]$ is already valid.
\end{itemize}


\noindent
\emph{Paper organization}.
Sect.~\ref{s:module} introduces the basics of the theory of reactive modules, which is reformulated using hypergraphs.
Sect.~\ref{s:hierarchical} describes a formalization of hierarchical modules.
Sect.~\ref{s:method} presents a proposed method that transforms hierarchical modules into decomposed forms.
Sect.~\ref{s:impl} describes a prototype implementation of the method and
Sect.~\ref{s:xp} reports an experimental result.
Sect.~\ref{s:related} describes the related work.

\vspace*{.5em}

\noindent
\emph{Preliminaries}.
We assume a basic knowledge of \emph{directed hypergraphs} $(V,E)$ where $V$ and $E$ are sets of vertices and \emph{(hyper)edges} (or hyperarcs), respectively.
Each hyperedge consists of two lists of vertices called the \emph{source} (or head) and the \emph{target} (or tail), respectively.
Given $e \in E$, we denote by $\In(e)$ and $\Out(e)$ the sets of the vertices in the source and target.
We call a vertex $v$ \emph{initial} if $\forall e \in E, v \notin \Out(e)$
and \emph{terminal} if $\forall e \in E, v \notin \In(e)$.
For details of the hypergraph theory, see e.g. \cite{ouvrardHypergraphs2020,brettoHypergraph2013}.

\section{Reactive Modules}
\label{s:module}

This section is a run-through introduction to the basics of the reactive module theory~\cite{alurReactive1999,alurSynchronous2015}, which is modified for our purpose.

We consider \emph{variables} typed as \lstinline|unit|, \lstinline|bool|, \lstinline|int|, 
etc.,
referring to the \emph{domains} $\Dom(\text{\lstinline{unit}}) = \{()\}$, $\Dom(\text{\lstinline{bool}}) = \{\True, \AB \False\}$, $\Dom(\text{\lstinline{int}}) = \mathbb{Z}$, 
etc.
Given a variable $v$ of type $t$, we denote its \emph{evaluation} by $\Eval{v} \in \Dom(t)$.
Given a set of variables $V = \{v_1,\ldots, v_n\}$ and a family of types $\{t_v\}_{v \in V}$, we denote the \emph{domain} 
\begin{equation*}
  \prod_{v \in V}\Dom(t_v) \quad \text{if $V \neq \emptyset$}, \qquad
  \{()\} \quad \text{if $V = \emptyset$},
\end{equation*}
by $\Dom(V)$.
%
%
%
For variable sets $V$ ($\neq \emptyset$) and $W \subseteq V$, and a subdomain $D = \mathrm{\Pi}_{v \in V} D_v \subseteq \Dom(V)$, $\pi_{W}(D)$ denotes the projection onto $W$ i.e. $\mathrm{\Pi}_{v \in W} D_v$.

\subsection{Task hypergraphs}

As a unit to describe a composite system, we consider stateless tasks that are non-blocking and can be nondeterministic.
\begin{definition}
  Let $R$ and $W$ be finite and mutually disjoint sets of variables.
  A \emph{task} $e$ with a \emph{read set} $R$ and a \emph{write set} $W$ represents a total relation in $\Dom(R) \times \Dom(W)$ such that $\forall r \!\in\! \Dom(R), \exists w \!\in\! \Dom(W), (r,w) \in e$.
  %
  We also denote a task by $e(R,W)$ to clarify its read and write sets.
\end{definition}
For example, a task $e(\emptyset,W)$ represents a task that outputs a constant or a nondeterministically chosen value in $\Dom(W)$.
Given  $e : \Dom(R) \to \Dom(W)$, it can represent a task that applies the function to the value of $R$ and writes it to $W$.

In this paper, we propose to formalize a composite task description as a hypergraph representing a network of tasks connected via read and write variables.
It aims to be an extension of task graphs in \cite{alurSynchronous2015} to depict relations between both tasks and variables.
\begin{definition} \label{d:tg}
  A \emph{task (hyper)graph} (TG) $(V,E)$ is a directed hypergraph whose vertices represent variables and hyperedges 
  represent tasks.
  We assume that 
  \emph{(i)}~$(V,E)$ is acyclic,
  \emph{(ii)}~no vertex is isolated,
  \emph{(iii)}~each vertex has at most one incoming edge, 
  \emph{(iv)}~for a task $e \in E$ such that $e \in \Dom(R) \Times \Dom(W)$ and a variable $v \in V$, 
  $v \in \In(e) \leftrightarrow v \in R$ and $v \in \Out(e) \leftrightarrow v \in W$.
\end{definition}
If it is clear from the context, we do not distinguish between vertices and variables, or hyperedges and tasks (relations in the variable domains), respectively.
\emph{Precedence relation} between tasks (denoted by $e \prec e'$ in \cite{alurSynchronous2015}) and \emph{await dependency} between variables ($v \succ v'$ in \cite{alurReactive1999,alurSynchronous2015}) are represented by the existence of a path between the two in the graph.
The condition~(iii) prevents conflicts between writes to a variable by multiple tasks.
\begin{example}
  Fig.~\ref{f:tg} illustrates an example TG.
  Each dot (with or without circle) represents a vertex in $\{i_1,i_2,o_1,o_2,s_1,s'_1,l_1\}$ and
  each set of directed lines mediated by a numbered circle represents a hyperedge in $\{e_1,e_2,e_3\}$.
  For instance, $\In(e_2) = \{i_2,s_1\}$ and $\Out(e_2) = \{l_1\}$.
  The hypergraph in Fig.~\ref{f:tg:abst} is not a TG since it contains a cycle.
\end{example}
TGs can be regarded as total relations i.e. tasks.
\begin{definition}
  Let $(V,E)$ be a TG, $R$ the set of initial vertices, and $W$ the set of vertices such that $W \subseteq V \setminus R$.
  We consider a relation
  \begin{equation} \label{eq:tg}
    \exists v_1 \!\in\! \Dom(t_1), \cdots \exists v_m \!\in\! \Dom(t_m), ~
    e_1(R_1,W_1) \land \cdots \land e_n(R_n,W_n),
  \end{equation}
  where $\{v_1,\ldots, v_m\} = V \setminus (R \cup W)$,
  $t_1$,\ldots, $t_m$ are their types, and $\{e_1,\ldots,e_n\} = E$.
  We denote the set of all such relations represented with a TG by $\TG(R,W)$.
\end{definition}
Note that variables in $R$ may always be the initial in the TG, but those in $W$ are not necessarily the terminal (they are shown as circled dots in the figures).
Every variable in $R_j$ or $W_j$ ($j = 1,\ldots,n$) not included in $R \cup W$ are bound by a quantifier in Eq.~\eqref{eq:tg}.

\begin{lemma}
  Every relation in $\TG(R,W)$ is total.
\end{lemma}
\begin{proof}
  The vertices in a TG can be partially ordered by the lengths of the longest paths from any initial vertex; we group the vertices according to the ordering.
  The initial vertices in $R$ belong to the first group.
  The vertices written by tasks with the empty read set belong to the second group.
  Then, we check 
  $\forall {u}_r \!\in\! \Dom(R), \exists {u}_w \!\in\! \Dom(W), \React({u}_r, {u}_w)$ 
  holds where $\React$ represents the relation~\eqref{eq:tg}.
  We check by induction that a value exists for vertices in every group to satisfy the relation.
  The first group is universally quantified, so any values can be assigned to the vertices.
  Assuming that the previous groups has been assigned values, 
  the values for the next group are determined by the incoming tasks.
  \qed
\end{proof}


Although tasks in \cite{alurSynchronous2015} are stateful, we formulate them as stateless.
Modules defined later specify the state variables among the read/write set of tasks and properly manage the states.
\emph{Atoms} are used instead of tasks in \cite{alurReactive1999}
to represent the initialization of the state of a module when executed and the reactions in each round. We embed the initial conditions in modules and represent only the reactions with TGs.

\subsection{Modules, implementation relation, and parallel composition}

Reactive and synchronous systems are formalized as compositional modules (called \emph{components} in \cite{alurSynchronous2015}) executed in a series of rounds.
\begin{definition}[\cite{alurReactive1999,alurSynchronous2015}]
A \emph{module} is a tuple $(I, O, S, \Init, \React)$ where
$I$, $O$ and $S$ are mutually disjoint sets of \emph{input}, \emph{output} and \emph{state} variables.
$\Init \subseteq \Dom(S)$ is an \emph{initial condition}, and
$\React$ is a \emph{reaction relation} that is a TG in $\TG(S \cup I, O \cup S')$, where $S'$ represents a set of variables renamed from $S$.
An 
\emph{execution} of a module is a 
sequence of reactions
\begin{multline*}
    s(-1) \xrightarrow{i(0)/o(0)} s(0) \xrightarrow{i(1)/o(1)} s(1)
    \cdots 
    ~~=~~ \\
    \{(s({j\!-\!1}), i(j), o(j), s(j)) \in \React \mid j \!\in\! \{0\} \cup \mathbb{N}, 
    s({-1}) \!\in\! \Init\}.
\end{multline*}
%
A \emph{trace} of an execution is 
a sequence of values $(i(0),o(0))\,(i(1),o(1)) \cdots$, 
i.e. the projection onto $I \cup O$.
\end{definition}
Note that $\React$ is interpreted as a set of quadruples of values, each of which is a family of values indexed by $S$, $I$, $O$ or $S'$.
We assume the state variables in $S$ and $S'$ always be the initial and terminal vertices of $\React$, respectively.
Differently from the formalization in \cite{alurSynchronous2015},
which embeds state variables within tasks,
modules must designate the vertices representing state variables from among the TG's initial and terminal vertices.
Hereafter, for a module $M_i$ with an identifier $i$, we denote its elements by e.g. $I_i$ and $\React_i$.

\begin{figure}[t!]
  \centering
  \begin{minipage}[b]{.45\textwidth}
    \centering
    \includegraphics[width=.65\textwidth]{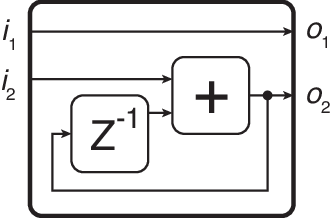} 
    \caption{A signal flow diagram describing the module $M_\Ex{\ref{e1}}$.}
    \label{f:sfd}
  \end{minipage}
  \begin{minipage}[b]{.45\textwidth}
    \centering
    \includegraphics[height=.12\textheight]{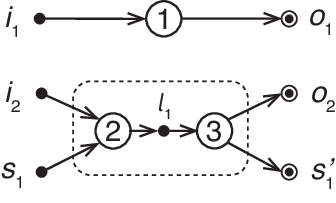} 
    \caption{The task hypergraph $\React_\Ex{\ref{e1}}$.} 
    \label{f:tg}
    \vspace{1em}
  \end{minipage}
  \vspace*{1em}

  \begin{minipage}[b]{.4\textwidth}
    \centering
    \includegraphics[height=.12\textheight]{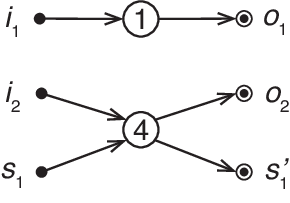} 
    \caption{Another TG for Ex.~\ref{e1}.} 
    \label{f:tg:simple}
  \end{minipage}
  \begin{minipage}[b]{.55\textwidth}
    \centering
    \includegraphics[width=.7\textwidth]{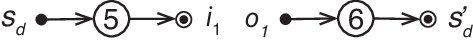}
    \caption{$\React_\mathrm{delay}$.} 
    \label{f:tg:delay}

    \vspace*{.5em}
    \includegraphics[width=.85\textwidth]{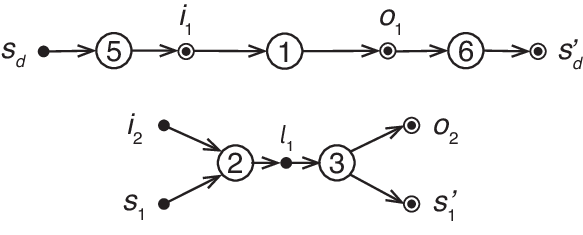}
    \caption{The TG of $M_\Ex{\ref{e1}} \PC M_\mathrm{delay}$.} 
    \label{f:tg:pc}
  \end{minipage}
\end{figure}

If $I$, $O$ or $S$ of a module is empty, there may be an execution that involves the value $()$.
For example, the module $M_\top$ defined by $(\emptyset,\emptyset,\emptyset,\{()\},\{((),(),(),())\})$ has an execution represented by a sequence of $((),(),(),())$.

As a graphical language to describe modules (e.g. Simulink), we consider \emph{signal flow diagrams} (SFDs).
An SFD consists of rectangles and directed lines 
annotated with variable names and other labels.
The rectangles represent modules that can be stateful and the lines represent synchronous communication between the processes.
Input signals are connected to the left side and output signals are extracted from the right side of rectangles.
\begin{example} \label{e1}
  The SFD of a counter constructed with an addition module and a delay module is shown in Fig.~\ref{f:sfd},
  which consists of 
  $I_\Ex{\ref{e1}} = \{i_{1},i_{2}\}$, $O_\Ex{\ref{e1}}=\{o_{1},o_{2}\}$, $S_\Ex{\ref{e1}}=\{s_{1}\}$, 
  $\Init_\Ex{\ref{e1}}(s_1) \equiv (\Eval{s_1}=0)$, and
  $\React_\Ex{\ref{e1}}$ represents a function $(\Eval{i_{1}},\Eval{i_{2}},\Eval{s_{1}}) \mapsto (\Eval{i_{1}},\Eval{i_{2}}+\Eval{s_{1}},\Eval{i_{2}}+\Eval{s_{1}})$.
  Fig.~\ref{f:tg} shows $\React_\Ex{\ref{e1}}$ as a TG.
  The vertices in $O_\Ex{\ref{e1}}$ are shown as circled dots and
  the vertex $l_1$ represents a local variable that does not belong to $I_\Ex{\ref{e1}}$, $O_\Ex{\ref{e1}}$ or $S_\Ex{\ref{e1}}$.
  The hyperedges $e_1$, $e_2$ and $e_3$ represent the identity, addition and copy functions.
  %
  An execution can be
  \begin{equation*}
      0 \xrightarrow{(\True,1)/(\True,1)} 
      1 \xrightarrow{(\False,0)/(\False,1)}
      1 \xrightarrow{(\True,2)/(\True,3)} 
      3 \cdots. 
  \end{equation*}
\end{example}
\begin{example}
  There can be multiple TGs representing a reaction relation.
  Fig.~\ref{f:tg:simple} shows another TG for $\React_\Ex{\ref{e1}}$
  in which $e_4$ outputs two copies of the sum.
\end{example}

We also represent a safety property of a coexisting module as a module.
A safety property $\varphi(v)$ involving a list of variables $v = (v_1,\ldots,v_n)$ is a module $M_{\varphi(v)}$ with $O_{\varphi(v)} = \{v_1,\ldots,v_n\}$, which nondeterministically outputs a signal satisfying $\varphi(v)$.
%
\begin{example}
  We can consider an invariance property $\Always\,o_2 \geq 0$ of $M_\Ex{\ref{e1}}$ described by the LTL.
  Output signals represented by variable $o_2$ of the module $M_\Ex{\ref{e1}}$ 
  must be mimicked by 
  $M_{\Always\,o_2 \geq 0}$.
\end{example}



Next, we introduce the implementation (or refinement) relation and the composition mechanism for the modules. 
\begin{definition}[\cite{alurReactive1999}]
  Let $M_1$ and $M_2$ be modules.
  We say $M_1$ \emph{implements} $M_2$, denoted by $M_1 \Impl M_2$, if
  \emph{(i)}~$O_2 \subseteq O_1$,
  \emph{(ii)}~$I_2 \subseteq I_1 \cup O_1$,
  \emph{(iii)}~the await dependency of $y \in O_2$ on $x \in I_2 \cup O_2$ in $\React_2$ (i.e. $y \succ x$) is preserved in $\React_1$, and
  \emph{(iv)}~for every trace $t$ of $M_1$, the projection of $t$ onto $I_2 \cup O_2$ is a trace of $M_2$.
  We denote 
  $M_1 \Impl M_2 \land M_2 \Impl M_1$ by $M_1 \cong M_2$.
\end{definition}

\begin{lemma}[\cite{alurReactive1999}]
  The implementation relation is a preorder.
\end{lemma}

\begin{definition}[\cite{alurReactive1999,alurSynchronous2015}]
  We say modules $M_1$ and $M_2$ are \emph{compatible} if  \emph{(i)}~$O_1 \cap O_2 = S_1 \cap S_2 = \emptyset$ and \emph{(ii)}~$\React_1 \cup \React_2$ (the union as graphs) is acyclic.
  %
  The \emph{parallel composition (PC)} $M_1 \PC M_2$ is a module $(I,O,S,\Init,\React)$, 
  where $I = (I_1 \cup I_2) \setminus O$, $O = O_1 \cup O_2$, $S = S_1 \cup S_2$, 
  $\Init = \Init_1\Times\Init_2$ and 
  $\React = \React_1 \cup \React_2$.
\end{definition}
\begin{example} \label{ex:pc}
  Consider to compose the module in Ex.~\ref{e1} with a module $M_\mathrm{delay}$ that represents a delay task whose TG is Fig.~\ref{f:tg:delay}.
  The TG of PC $M_\Ex{\ref{e1}} \PC M_\mathrm{delay}$ is shown in Fig.~\ref{f:tg:pc}.
\end{example}
\begin{lemma}[\cite{alurReactive1999}]
  The PC operation is 
  associative, transitive and symmetric.
\end{lemma}


To deal with extra output variables added by the PC operation, we introduce the {hiding} operator.
\begin{definition}[\cite{alurReactive1999,alurSynchronous2015}] \label{d:hiding}
  Given a module $M$ and a subset of output variables $O' \subseteq O$, \emph{hiding} of $O'$ in $M$, denoted $M \setminus O'$, is a module consisting of the same elements as $M$ but excluding $O'$ from $O$.
\end{definition}

\subsection{Compositional verification}

We consider to verify that a module $M$ fulfils an \emph{assume-guarantee contract} $(M_a,M_g)$, i.e. a pair of modules.
For that purpose, we can show $M \PC M_a \Impl M_g$ holds.
In our experiment in Sect.~\ref{s:xp}, we consider contracts consisting of safety properties.
\begin{example} \label{ex:impl}
  $M_\Ex{\ref{e1}} \PC M_{\Always i_2 \geq 0} \Impl M_{\Always o_2 \geq 0}$ holds.
\end{example}

In this paper, we consider the verification of systems composed of $n$ submodules.
Here, we assume that each submodule satisfies a given contract, and
aim at efficiently verifying the fact that the entire system fulfils the top-level contract by utilizing the assumptions.
\begin{definition} \label{d:cv}
  A \emph{compositional verification problem} consists of 
  $n$ modules $M_1$, \ldots, $M_n$, a top-level contract $(M_a,M_g)$ and $n$ sub-contracts $(M_{a.j}, M_{g.j})$ 
  where $j = 1,\ldots,n$;
  we assume $M_j \PC M_{a.j} \Impl M_{g.j}$ for every $j$.
  The \emph{goal} is a condition $M_1 \PC \cdots \PC M_n \PC M_a \Impl M_g$.
\end{definition}
The module $M_\top$ can be used to omit some elements of contracts.
For arbitrary $M$, $M \Impl M_\top$ and $M \PC M_\top \cong M$ hold.

The following lemma provides two basic inference rules for the compositional reasoning 
on modules.
The second rule allows a PC $M_1 \PC M_2$ to be \emph{circular} i.e. $I_1 \cap O_2$ and $I_2 \cap O_1$ are nonempty.
\begin{lemma}[\cite{alurReactive1999}] \label{th:cr}
  Let $M_1$, $M_2$, $M_3$ and $M_4$ be modules, where 
  $M_1$ and $M_2$, and $M_3$ and $M_4$ are respectively compatible,
  and $I_3 \cup I_4 \subseteq I_1 \cup I_2 \cup O_1 \cup O_2$.
  \emph{(i)}~$M_1 \PC M_2 \Impl M_1$.
  \emph{(ii)}~If $M_1 \PC M_3 \Impl M_4$ and $M_2 \PC M_4 \Impl M_3$, then $M_1 \PC M_2 \Impl M_3 \PC M_4$.
\end{lemma}

\begin{example} \label{ex:impl:infer}
  We consider a compositional verification problem of the system
  $M_\Ex{\ref{e1}} \PC M_\mathrm{delay}$ in Ex.~\ref{ex:pc}, which consists of 
  \begin{itemize}
    \item submodules $M_\Ex{\ref{e1}}$ and $M_\mathrm{delay}$,
    \item top-level contract $(M_\top, M_{\Always o_2 \geq 0})$, and 
    \item sub-contracts $(M_{\Always i_2 \geq 0},M_{\Always o_2 \geq 0})$ and $(M_{\Always o_2 \geq 0},M_{\Always i_2 \geq 0})$.
  \end{itemize}
  The goal $M_\Ex{\ref{e1}} \PC M_\mathrm{delay} \PC M_\top \Impl M_{\Always o_2 \geq 0}$ is provable.
  First, we can deduce as follows.
  \vspace*{-1em}
  \begin{prooftree} \small
    \AxiomC{$M_{\Ex{\ref{e1}}} \PC M_{\Always i_2 \geq 0} \Impl M_{\Always o_2 \geq 0}$}
    \noLine
    \def\extraVskip{1pt}
    \UnaryInfC{$M_\mathrm{delay} \PC M_{\Always o_2 \geq 0} \Impl M_{\Always i_2 \geq 0}$}
    \RightLabel{\scriptsize Lem.\ref{th:cr}(ii)}
    \UnaryInfC{$M_{\Ex{\ref{e1}}} \PC M_\mathrm{delay} \Impl M_{\Always i_2 \geq 0} \PC M_{\Always o_2 \geq 0}$}
    \AxiomC{}
    \RightLabel{\scriptsize Lem.\ref{th:cr}(i)}
    \UnaryInfC{$M_{\Always i_2 \geq 0} \PC M_{\Always o_2 \geq 0} \Impl M_{\Always o_2 \geq 0}$}
    \RightLabel{\scriptsize Trans.}
    \BinaryInfC{$M_{\Ex{\ref{e1}}} \PC M_\mathrm{delay} \Impl M_{\Always o_2 \geq 0}$}
  \end{prooftree}
  Then, the goal follows from $M_\Ex{\ref{e1}} \PC M_\mathrm{delay} \PC M_\top \Impl M_\Ex{\ref{e1}} \PC M_\mathrm{delay}$.
\end{example}

\section{Hierarchical Reactive Modules}
\label{s:hierarchical}

\begin{figure}[t!]
  \centering
  \begin{minipage}[b]{.4\textwidth}
    \centering
    \includegraphics[height=.12\textheight]{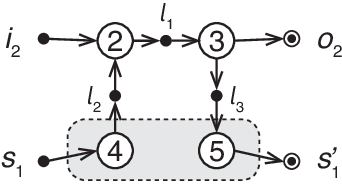} 
    \caption{A detailed TG for Ex.~\ref{e1}.}
    \label{f:tg:detailed}
  \end{minipage}
  \begin{minipage}[b]{.55\textwidth}
    \centering
    \subfloat[Sub(hyper)graph.\label{f:tg:sub}]{\includegraphics[height=.06\textheight]{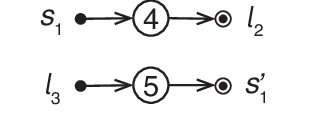}}
    \hspace*{.5em}
    \subfloat[Abstraction.\label{f:tg:abst}]{\includegraphics[height=.1\textheight]{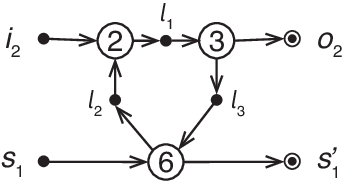}}
    \caption{Separation of a hypergraph.} 
    \label{f:tg:sep}
  \end{minipage}
\end{figure}

Practical languages for describing modules (e.g. Lustre and Simulink) tend to support hierarchical structures.
We formalize such structures by allowing a TG to be hierarchical, separating a subhypergraph of the TG as a submodule.
\begin{definition}
  For a TG $(V,E)$,
  its \emph{sub(hyper)graph} is a TG such that $V' \subseteq V$ and $E' \subseteq E$.
  %
  Assume a subgraph $(V',E')$ of $(V,E)$ is in $\TG(R,W)$.
  Then, the \emph{abstraction of $(V',E')$ in $(V,E)$}
  is a hypergraph $(V'',E'')$ where $V''$ is a set such that $V' \cup V'' = V$ and 
  $E'' = (E \setminus E') \cup \{e\}$ where $e$ is a fresh hyperedge such that $\In(e) = R$ and $\Out(e) = W$.
\end{definition}
According to \cite{brettoHypergraph2013}, subgraphs defined above are partial subhypergraphs without isolated vertex.
Note that an abstraction of a TG 
may contain a cycle as shown in the following example.
\begin{example} \label{ex:sub}
  Consider a hypergraph in Fig.~\ref{f:tg:detailed} that describes yet another TG (lower part) for Ex.~\ref{e1}.
  The graph in Fig.~\ref{f:tg:sub} is its subgraph that consists of the separated edges $e_4$ and $e_5$.
  The abstraction of the subgraph in Fig.~\ref{f:tg:detailed} is shown in Fig.~\ref{f:tg:abst}, in which
  $e_6$ is associated with the subgraph.
\end{example}
%


Now, we consider the TGs of modules to be hierarchical.
Intuitively, a hierarchical module is a module that separates subgraphs as submodules. To ensure the consistency among decomposed descriptions, and because of the proposed method, we assume several conditions.
\begin{definition} \label{d:hierarchical}
  Let $M_1,\ldots,M_n$ be modules. 
  A \emph{hierarchical} module $M[M_1,\ldots, \AB M_n]$ (also denoted by $M$ or $M[M_1..M_n]$) is a module that satisfies 
  %
  the following conditions for each submodule $M_j$:
  \emph{(i)}~$I \cap I_j = O \cap O_j = \emptyset$,
  \emph{(ii)}~$S_j \subseteq S$,
  \emph{(iii)}~$\pi_{S_j}(\Init) \equiv \Init_j$, and
  \emph{(iv)}~$\React_j$ is a subgraph of $\React$.
  %
\end{definition}
The condition~(i) forces a submodule to handle its own input and output variables for the ease of the proposed method in Sect.~\ref{s:method}.
The input and output variables must be copied to local variables before communicating with submodules.
The conditions~(ii) and (iii) make the state variables of a submodule shared with the parent and are maintained properly.
%
A hierarchical module is interpreted as a \emph{flattened} module whose TG embeds the submodules' TGs.
In order for a hierarchical module $M[M_1..M_n]$ and its submodules $M_1$,\ldots, $M_n$ to be interpreted as modules properly,
the TGs of the submodules must avoid a write conflict (Def.~\ref{d:tg}) and their states must be managed with separate variables,
i.e., $\bigcap_{j=1,\ldots,n} O_j = \AB \bigcap_{j=1,\ldots,n} S_j \AB = \emptyset$ must hold.

\begin{example} \label{e2}
  We consider a hierarchical module $M_\Ex{\ref{e2}}[M_1,M_2]$ that has two counter modules as submodules (Fig.~\ref{f:sfd:hierarchical}).
  We assume they are the same as in Ex.~\ref{e1}, except for the variable names.
  Its TG is shown in Fig.~\ref{f:tg:hierarchical} where the dashed frames enclose the hyperedges in the subgraphs.
  The state variables of the submodules are inherited to $M_\Ex{\ref{e2}}$ ($S_\Ex{\ref{e2}} = \{s_1,s_{1.1},s_{2.1}\}$) and $\Init_\Ex{\ref{e2}}$ is set as 
  $\{\top\}\Times\Init_1\Times\Init_2$.
  %
  $M_\Ex{\ref{e2}} \PC M_{\Always i_1 \geq 0} \Impl M_{\Always o_1 \geq 0}$ holds.
\end{example}

\begin{figure}[t!]
  \centering
  \includegraphics[width=.75\textwidth]{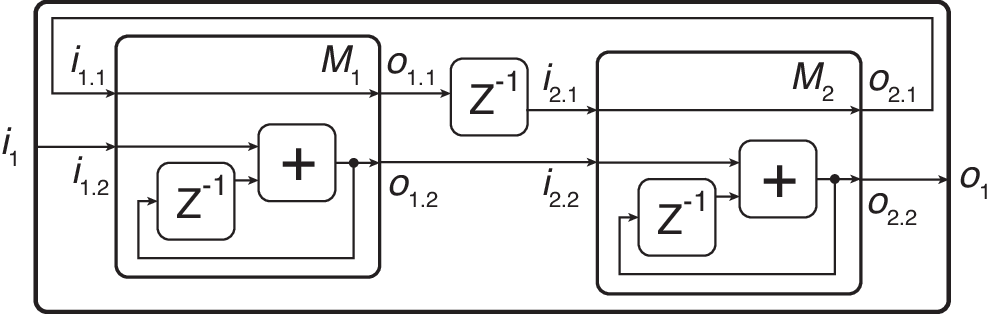} 
  \caption{A hierarchical module $M_\Ex{\ref{e2}}$.}
  \label{f:sfd:hierarchical}
  \vspace*{-7em}
  \includegraphics[width=.9\textwidth]{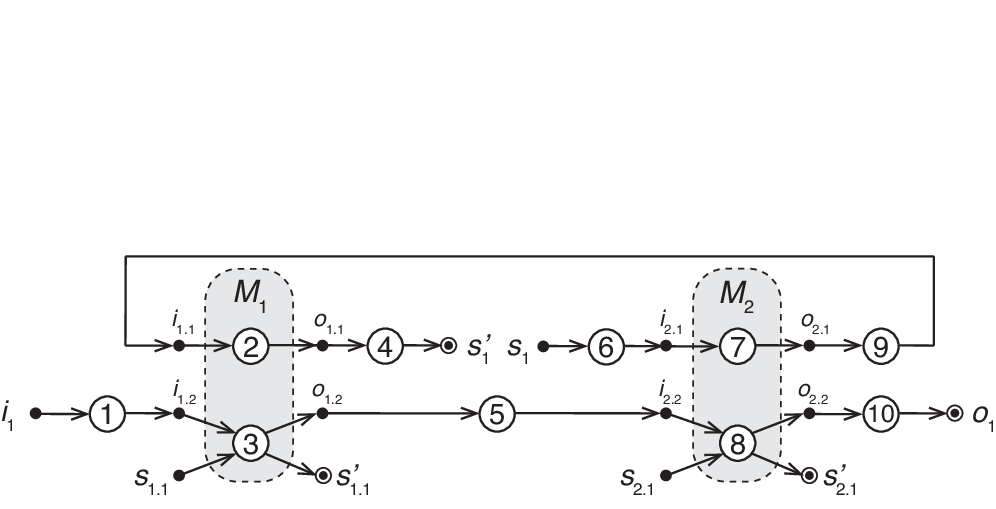} 
  \caption{The TG of $M_\Ex{\ref{e2}}$.}
  \label{f:tg:hierarchical}
\end{figure}



%

\subsubsection{Abstraction of hierarchical modules.}
The hierarchization of modules can be viewed as an abstraction as shown in Ex.~\ref{ex:sub} (Fig.~\ref{f:tg:abst}).
Naively, we can replace a fragment of the TG that belongs to a submodule with a hyperedge, which is then regarded as a complete graph between input and output vertices. Then, it may help to efficiently search for a counterexample that can be executed by the parent module.
The Kind2 model checker~\cite{championKIND2016} abstracts each submodule $M_j$ with a hyperedge representing a property 
$\Hist\,\varphi(i_j) \Rightarrow \psi(o_j)$, 
where $\Hist$ is the pLTL ``historically'' modality, to exploit the contract given as a pair of safety properties $(\varphi(i_j), \psi(o_j))$ (where $i_j$/$o_j$ is a variable list in $I_j$/$O_j$).
Separately, the fact $M_j \PC M_{\varphi(i_j)} \Impl M_{\psi(o_j)}$ has to be verified to check that the submodule $M_j$ satisfies the contract.

\begin{example}
  Let $\varphi(x) \equiv \Always\,x \geq 0$.
  For the submodules $M_j$ of $M_\Ex{\ref{e2}}$ ($j = 1,2$), $M_j \PC M_{\varphi(i_{j.2})} \Impl M_{\varphi(o_{j.2})}$ holds.
  Then, the fact $M_\Ex{\ref{e2}} \PC M_{\varphi(i_1)} \Impl M_{\varphi(o_1)}$ can be verified with an abstraction that replaces each submodule $M_j$ with $M_{\Hist\,\varphi(i_{j.2}) \Rightarrow \varphi(o_{j.2})}$.
\end{example}

While the abstraction methods enable a verification process that leverages the sub-contracts,
they cannot properly handle {circular} systems whose submodules assume the existence of other submodules.
In such cases, 
the abstraction approach might not work;
the process may detect spurious counterexample.

\begin{example} \label{ex:spurious}
  Consider the verification of $\varphi(x) \equiv \Always\,x \geq 0$ on $M_\Ex{\ref{e2}}$ again.
  Let $\psi(x_1,x_2) \equiv \Always(x_1 \land x_2 \geq 0)$.
  For the submodules $M_j$ ($j = 1,2$), $M_j \PC M_{\psi(i_{j.1},i_{j.2})} \Impl M_{\psi(o_{j.1},o_{j.2})}$ holds.
  So, we can abstract the subgraph for $M_j$ with $\Hist\,\psi(i_{j.1},i_{j.2}) \Rightarrow \AB \psi(o_{j.1},o_{j.2})$ in the TG of $M_\Ex{\ref{e2}}$. 
  %
  Then, the verification 
  of $M_\Ex{\ref{e2}} \PC M_{\varphi(i_1)} \Impl M_{\varphi(o_1)}$ 
  will result in a false counterexample such that $\Eval{i_{1.1}} = \Eval{i_{2.1}} = \False$.
\end{example}

\section{Compositional Verification of Hierarchical Modules}
\label{s:method}

In this section, we propose a compositional verification method that can handle circular hierarchical modules.
The method validates that a module $M[M_1..M_n]$ satisfies a contract $(M_a,M_g)$.
We describe how to validate the goal under the assumption 
that every submodule $M_j$ ($j \in \{1,\ldots,n\}$) satisfies its contract $(M_{a.j},M_{g.j})$.
The key idea is to prepare a module $M^\dagger$ called \emph{adapter} by extracting only the top-level part of the hierarchical TG.
The subgraphs are separated from the top-level TG and the variables that correspond to the boundary vertices between the top-level part and the subgraphs are set as input and output variables of the adapter module.
The adapter $M^\dagger$ is prepared in a way $M[M_1..M_n]$ and $M_1 \PC \cdots \PC M_n \PC M^\dagger$ be isomorphic, and thus yields a compositional verification problem (Def.~\ref{d:cv}).

\begin{definition} \label{d:adapt}
  For sets $V_1,\ldots,V_n$, we denote their union by $V_{\{1..n\}}$.
  The \emph{adapter} $\Adapt{M}$ of a hierarchical module $M[M_1..M_n]$
  is a module 
  with the components
  $\Adapt{I} = I \cup O_{\{1..n\}}$, 
  $\Adapt{O} = O \cup I_{\{1..n\}}$, 
  $\Adapt{S} = S \setminus S_{\{1..n\}}$,
  $\Adapt{\Init} = \pi_{\Adapt{S}}(\Init)$, and
  $\Adapt{\React}$ obtained from $\React$ by removing the hyperedges $e_j$ and the vertices corresponded with the variables in $S_j$ for every $j$.
\end{definition}

\begin{figure}[t!]
  \centering
  \includegraphics[width=.75\textwidth]{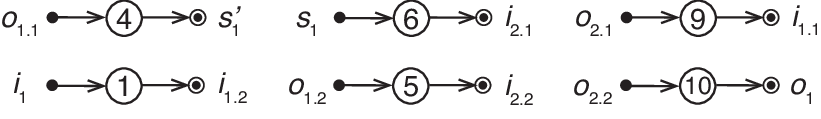} 
  \caption{The TG of $\Adapt{M}_\Ex{\ref{e2}}$.}
  \label{f:tg:adapter}
\end{figure}

\begin{example}
  The TG of $\Adapt{M_\Ex{\ref{e2}}}$ is with 
  $R = \{i_1, o_{1.1},o_{1.2},o_{2.1}, \AB o_{2.2}, \AB s_1\}$ and 
  $W = \{o_1, i_{1.1},i_{1.2}, \AB i_{2.1},i_{2.2}, s'_1\}$.
  It is illustrated as Fig.~\ref{f:tg:adapter}.
\end{example}

A hierarchical module can now be regarded as a PC of decomposed modules. 
However, since extra output variables of the submodules are added by the PC, they must be hidden (Def.~\ref{d:hiding}).
\begin{lemma} \label{th:adapt}
  Let $M[M_1..M_n]$ be a hierarchical module and $M'$ be $M_1 \PC \cdots \PC M_n \PC \AB \Adapt{M}$. 
  Then, $M[M_1..M_n] \cong M' \setminus (I_{\{1..n\}} \cup O_{\{1..n\}})$.
\end{lemma}
\begin{proof}
  In Def.~\ref{d:adapt}, the variable sets $I_{\{1..n\}}$ and $O_{\{1..n\}}$ are added to the adapter as initial and terminal vertices in $\Adapt{\React}$.
  Then, the PC of $M_j$ and $\Adapt{M}$ merges $\React_j$ and $\Adapt{\React}$ by matching the vertices for the variables in $I_j \cup O_j$.
  Hence, the resulting TG is equivalent of replacing $e_j$ with $\React_j$ in Def.~\ref{d:hierarchical}~(ii).
  %
  Each set of the variables is equal to that of $M[M_1..M_n]$ by the hiding operation.
  \qed
\end{proof}

Once a hierarchical module is decomposed into a PC of submodules and an adapter,
it is possible to perform a compositional verification of facts about the module
based on the theory of reactive modules.
To do so, we can make a proof using the inference rules in Sect.~\ref{s:module}, which may require a manual effort.
However, the following theorem shows that such verification always succeed, 
if the conditions on the submodules and the adapter are valid.
\begin{theorem} \label{t:main}
  Consider a goal $M[M_1..M_n] \,\PC\, M_a \Impl M_g$.
  If \emph{(a)}~$M_j \PC M_{a.j} \AB \Impl M_{g.j}$ for every $j$ and \emph{(b)} $\Adapt{M} \PC M_a \PC M_{g.1} \PC \cdots \PC M_{g.n} \Impl M_g \PC M_{a.1} \PC \cdots \AB \PC M_{a.n}$,
  then the goal holds.
\end{theorem}
\begin{proof}
  Let $M_{[j..k]}$ represent $M_j \PC \cdots \PC M_k$ if $j \leq k$ and
  the empty module $M_\top$ if $j > k$ (we use the same notation for $M_{a.[j..k]}$ and $M_{g.[j..k]}$).
  We rewrite the condition~(b) as follows (we assume $j=n$ initially):
  \begin{equation} \label{eq:bmod}
    M_{[(j+1)..n]} \PC \Adapt{M} \PC M_a \PC M_{g.[1..j]} \Impl M_g \PC M_{a.[1..j]}.
  \end{equation}
  From (a), $(M_j \PC M_{a.j}) \PC (M_{a.[1..(j-1)]} \PC M_g) \Impl M_j \PC M_{a.j}$ (Lem.~\ref{th:cr}~(i)), and the transitivity of $\Impl$, we have
  \begin{equation} \label{eq:amod}
    M_j \PC M_g \PC M_{a.[1..j]} \Impl M_{g.j}.
  \end{equation}
  From \eqref{eq:bmod}, \eqref{eq:amod} and Lem.~\ref{th:cr}~(ii), we have
  \begin{equation*}
    M_{[j..n]} \PC \Adapt{M} \PC M_a \PC M_{g.[1..(j-1)]} \Impl M_g \PC M_{a.[1..j]} \PC M_{g.j}.
  \end{equation*}
  The rhs can be simplified to $M_g \PC M_{a.[1..(j-1)]}$ (by Lem.~\ref{th:cr}~(i) and the transitivity).
  By repeating the above for $j=n-1,\ldots,1$, we will obtain the fact, which is equivalent to the goal due to Lem.~\ref{th:adapt}.
  \qed
\end{proof}

Proof of a hierarchical module may be obtained by analyzing the content of the top-level as a nested PC of submodules.
However, the proof process for a verification goal is non-trivial in general. 
Applying the rules in Lem.~\ref{th:cr} requires a number of deductions, introduction of PCs of submodules and property modules, and adjusting the form of PC terms while checking module compatibility.
The proposed method transfers the analysis of the compositional structure of the top-level to the verification process of the condition~(b).

When a system contains multiple hierarchies, we can simply repeat the process to check the two conditions in Th.~\ref{t:main} in a bottom-up fashion to verify the whole system.
Each process for a hierarchical module decomposes it after deriving an adapter, then either verifies the contract for the PC of the submodules or composes them with other submodules of the parent.

\section{Implementation}
\label{s:impl}

We have implemented the proposed method by extending the Kind2 tool.

\subsection{Lustre, CoCoSpec, and Kind2}

\emph{Kind2} (version~1.6.0)\cite{championKIND2016} is an SMT-based model checking tool (we used Z3~4.12.4 as an SMT solver).
Its input language is \emph{Lustre}~\cite{caspiLUSTRE1987}, 
a textual language for describing hierarchical synchronous modules, and the modules can be annotated contracts with the \emph{CoCoSpec} language~\cite{championCoCoSpec2016}.
\begin{example} \label{e3}
Fig.~\ref{f:lustre} shows an example described with Lustre and CoCoSpec.
It is a module similar to Ex.~\ref{e2} in which the counters are replaced with second-order digital filters.
A Lustre node is defined by a section started with \lstinline|node|,
which is followed by 
\begin{itemize}
  \item the node name (e.g. \lstinline|Filter| and \lstinline|Toplevel|), 
  \item the input variable list (in parentheses), 
  \item the output variable list (with keyword \lstinline|returns|), 
  \item the contract annotation (described within a comment), 
  \item the local variable list (with keyword \lstinline|var|), and 
  \item the body (enclosed in \lstinline|let| and \lstinline|tel|).
  The line ``\lstinline|--%MAIN;|'' specifies that the node is a verification target.
\end{itemize}
We consider modules to be instances of Lustre nodes whose input and output variables are substituted by the arguments.
The above program describes the verification conditions 
$M_{\texttt{F}0} \PC M_{\texttt{F}0.a} \Impl M_{\texttt{F}0.g}$ and
$M_{\texttt{T}0}[M_{\texttt{F}1},M_{\texttt{F}2}] \PC M_{a.\texttt{T}0} \Impl M_{g.\texttt{T}0}$
where $M_{\texttt{N}i}$ represents the $i$th instance of the node \texttt{N} and $M_{a.\texttt{N}i}$ and $M_{g.\texttt{N}i}$ represent the annotated properties (we abbreviate \texttt{Filter} and \texttt{Toplevel} as \texttt{F} and \texttt{T}).
\end{example}

\begin{figure}[t!]
  \lstset{frame=single}
  \lstset{numbers=left}
  \begin{lstlisting}[basicstyle=\ttfamily\footnotesize, columns=flexible, keepspaces=true]
node Filter (in1 : bool; in2 : real) 
returns (out1 : bool; out2 : real);
(*@contract
  assume    in1;  assume    -1.0 <= in2 and in2 <= 1.0;
  guarantee out1; guarantee -1.0 <= out2 and out2 <= 1.0;
*)
var sum, D1, D2: real;
let
  out1 = in1;
  sum = 0.0582*(if in1 then in2 else -in2) - (-1.49*D1) - 0.881*D2;
  D1 = 0.0 -> pre sum;      D2 = 0.0 -> pre D1;
  out2 = (sum - D2) / 1.25;
tel

node Toplevel (in : real) returns (out : real);
(*@contract
  assume    -1.0 <= in and in <= 1.0;
  guarantee -1.0 <= out and out <= 1.0;
*)
var b1, b2, pre_b2 : bool; s1 : real;
let
  b1, s1 = Filter(b2, in);  pre_b1 = true -> pre b1;
  b2, out = Filter(pre_b1, s1);
  --%MAIN;
tel
\end{lstlisting}
\caption{An example Lustre program annotated with CoCoSpec.}
\label{f:lustre}
\end{figure}

When an input describes a goal (verification condition for the target module) $M[M_1..M_n] \PC M_a \Impl M_g$ and conditions for submodules $M_j \PC M_{a.j} \Impl M_{g.j}$ where $j = 1,\ldots,n$, Kind2 is able to verify its validity in three modes:
\begin{itemize}
  \item \emph{Monolithic mode} that interprets the target $M[M_1..M_n]$ as a module (cf. Def.~\ref{d:hierarchical}) and verifies only the goal.
  \item \emph{Modular mode} that verifies only the conditions for submodules $M_1$, \ldots, $M_n$. 
  \item \emph{Compositional mode} that verifies the goal with an abstraction as described in Sect.~\ref{s:hierarchical}.
\end{itemize}
We used Kind2 with the default setting, which runs several model checking algorithms e.g. BMC, $k$-induction and PDR, in parallel.

\subsection{Implementation of the proposed method}

We have implemented the proposed method in OCaml by modifying Kind2. 
The function we have added translates a hierarchical Lustre program to a program in which the hierarchical modules are replaced with adapter modules (here, we also refer to Lustre \emph{nodes} as {modules}).
The input is Lustre programs annotated with CoCoSpec contracts. 
%
It generates a list of reactive modules by applying the following processes:
\begin{enumerate}
  \item \emph{Instantiation of module definitions}.
    Because a Lustre node may be invoked several times by the parent module, we instantiate a node definition with the real argument of each call.
  \item \emph{Modification of the top-level into an adapter}.
    For each hierarchical module, we remove the submodule invocation statements and modify the variable list as described in Def.~\ref{d:adapt}.
  \item \emph{Pretty printing}. 
    The intermediate data will be printed as a decomposed and properly annotated Lustre program. 
\end{enumerate}
The source code is available at \url{https://github.com/dsksh/kind2}.

By feeding the output Lustre program to Kind2 with the {modular mode}, the satisfaction of the assume-guarantee contract by each module will be checked.
The success of the process implies the validity of the annotated top-level module in the original Lustre program.

\section{Experiment}
\label{s:xp}

To evaluate the effectiveness and the performance of the proposed method, 
we have conducted the verification of several examples using the implementation.
The experiment was done with a MacBook Pro (10-core Apple M2 Pro chip and 32GB RAM).


We prepared several circular hierarchical modules for the experiment.
\begin{itemize}
  %
  \item \emph{Feedback loop system containing $n$ digital filters} ($n$Filters).
  This is an extended and parameterized version of Ex.~\ref{e3}.
  The \lstinline|Toplevel| instance has $n$ \lstinline|Filter| instances and they form a loop as illustrated in Fig.~\ref{f:ex:nf}.
  We gave the same assume-guarantee contracts to \lstinline|Filter| and \lstinline|Toplevel| as in Ex.~\ref{e3}, and verified that \lstinline|Toplevel| satisfies the contract.
  \item \emph{DC motor control (MCtrl)}.
  This is a more practical and typical example in which a motor and a controller are described as submodules $M_2$ and $M_1$ and form a feedback loop as illustrated in Fig.~\ref{f:ex:mc}.
  We annotated the top-level module with 8 safety (guarantee) properties,
  and submodules with 3 to 5 assume/guarantee properties.
\end{itemize}

\begin{figure}[t!]
  \centering
  \subfloat[$n$Filters.\label{f:ex:nf}]{\includegraphics[height=.075\textheight]{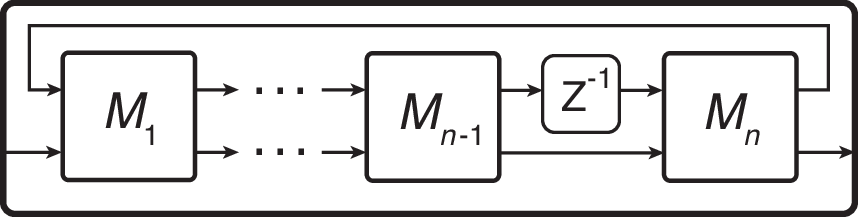}}
  \hspace*{4em}
  \subfloat[MCtrl.\label{f:ex:mc}]{\includegraphics[height=.075\textheight]{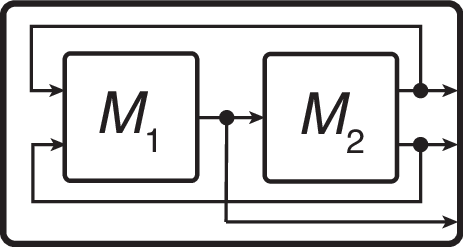}}
  \caption{SFDs of the example modules.} 
  \label{f:ex}
\end{figure}

\begin{table}[t]
    \centering
    \small
    \caption{Execution result.} \label{t:result}
    \setlength{\tabcolsep}{10pt}
    \begin{tabular}{l|rrr}
        \toprule
        Example & Time (mono.) & Time (comp.) & \#Guarantees \\
        \midrule
        $2\text{Filters}$  & 204s & 14.9s & 7  \\
        $3\text{Filters}$  & TO   & 14.9s & 9  \\
        $36\text{Filters}$ & TO   & 15.4s & 75 \\
        MCtrl & TO & 3.1s & 24 \\
        \bottomrule
    \end{tabular}
\end{table}

We first verified the target hierarchical modules with the monolithic mode of Kind2.
Second, we decomposed the target modules into the PC forms using the proposed method, and then verified that the submodules and the adapters satisfied their contracts (i.e. the conditions in Th.~\ref{th:cr}) using Kind2.
Note that, if we verify any of the examples with the compositional mode of Kind2, it will result in a spurious counterexample as explained in Ex.~\ref{ex:spurious}.

\vspace*{-1em}

\subsubsection{Experimental result.}

The result is shown in Table~\ref{t:result}.
Each column shows wall clock time for the monolithic process (``mono.'') or the process with the proposed method (``comp.''), and the total number of guarantee properties verified with the proposed method (``\#Guarantees'').
``TO'' means the process did not terminate within 600s.

\vspace*{-1em}

\subsubsection{Discussions.}

Since the submodules i.e. the digital filters, motor and controller behave in a stateful manner,
using Kind2 to check the safety properties requires analyzing the execution prefixes of certain lengths, which would be time-consuming.
Therefore, the verification of the system at once resulted in timeouts except for 2Filters.
%
When the proposed method verified the examples by dividing them into modules, the process was more efficient because the numbers of rounds analyzed were reduced.

In each experiment for $n$Filters, the proposed method verified only a \lstinline|Filter| instance because $n$ verification conditions for submodules were for the same Lustre node.
Therefore, differences in the value of $n$ should appear only in the verification of adapter module;
the number of read/write variables of the module and the number of corresponding guarantee properties would increase.
The execution time increased slightly;
although this was due to the simplicity of the top-level content,
we consider the overhead of our method was small since increasing the number of variables did not have much effect.
The effectiveness of our method was also confirmed in MCtrl.

%

\section{Related Work}
\label{s:related}

The hierarchy with nesting parallel compositions has been considered in an implementation~\cite{alurMOCHA1998} and extensions e.g. \cite{alurModular2000} of the reactive module theory.
The hierarchization we consider in this paper is slightly different from the nesting of composition operations; ours corresponds to the embedding operation in dataflow diagrams.
Alur et al.~\cite{alurCompositional2006} address two kinds of hierarchies by agents and modules, and propose to perform reasoning on them separately; they do not consider a transformation between the two unlike this work.

More recently,
there has been work on compositional verification of hierarchical Simulink models~\cite{murugesanCompositional2013,bostromContractbased2016,dragomirCompositional2016,dragomirRefinement2020}.
As in this paper, these methods give contracts to subsystems and verify the models according to a hierarchical structure.
Bostr\"om and Wiik~\cite{bostromContractbased2016} propose to convert both models and contracts to specific dataflow graphs, then to sequential programs, and perform program verification. 
Their method does not support models with algebraic loops, and it is not clear whether it can handle the circular systems we consider.
Dragomir et al.~\cite{dragomirRefinement2020,preoteasarefinement2022} 
introduce QLTL properties and dedicated refinement calculus to perform verification on hierarchical models interactively on Isabelle.
Notably, they handle liveness properties which we do not.
Their verification method requires human support unlike ours.
%
Since they do not seem to provide any inference rules that explicitly deal with circular cases, it is unclear whether our method is applicable to their framework.

The Kind2 tool~\cite{championKIND2016} supports the modular and compositional model checking of hierarchical Lustre programs as described in Sect.~\ref{s:hierarchical} and Sect.~\ref{s:impl}. It is limited in handling circular programs.

Dragomir et al.~\cite{dragomirCompositional2016} propose to analyze hierarchical models with composite predicate transformers (CPTs) that use several composition operators, e.g. serial and parallel compositions and feedbacks.
In comparison, we use only one composition operator and formalize the hierarchical structure in a fixed way.
Tripakis et al.~\cite{tripakisModular2018} formalize hierarchical dataflow diagrams and propose a profiling method to assure the modularity; although the subject is similar, their purpose is different from ours.
Bakirtzis et al.~\cite{bakirtzisCategorical2021} propose a general framework for various compositional CPS models, which includes concepts equivalent to modules and hierarchies and a formalization of contracts;
however, they do not discuss either a transformation between the two concepts or verification methods.
Fong and Spivak~\cite{fongInvitation2019} formalize various hierarchical models of reactive systems, but do not consider synchronous behavior or assume-guarantee verification.

Automation of contract generation has been studied, e.g. \cite{pasareanuLearning2008,abdelkaderAutomated2018,neeleCompositional2023}.
We assume that contracts are given; the combination of our method and contract  generation is a future issue.


%

\section{Conclusion}

We have formalized hierarchical synchronous systems based on the theory of reactive modules.
We have then proposed a verification method that decomposes a hierarchical module into non-hierarchical modules and checks each module separately to show that the whole system satisfies the contract.
As the experimental results show, the proposed method can effectively verify the systems with circular structures, which are suitable for describing plant control CPSs.

In general, compositional reasoning requires proof of the consistency among the verification results of each module, but this is not necessary when a top-level module is decomposed with our method.
The proof task to analyze the system structure is delegated to the implementation relation on the adapter, and then it is efficiently discharged using a tool like Kind2.

Future work includes integrating the proposed method with other compositional methods such as for triggered modules and different-rate modules.
Also, cooperation with automatic contract generation methods remains an issue.

\subsubsection*{Acknowledgements.}

Tien Duc Ngo contributed to the preliminary experiment of this work.
This work was supported by JST, CREST Grant Number JPMJCR23M1 and 
JSPS, KAKENHI Grant Number 23K11969.

\bibliographystyle{splncs04}
\bibliography{cv}


%

%


\end{document}